\documentclass[a4paper,11pt]{article}

\usepackage{fullpage}
\usepackage{latexsym}
\usepackage{graphicx}
\usepackage{amsthm}

\newtheorem{theorem}{Theorem}[section]
\newtheorem{lemma}[theorem]{Lemma}

\begin{document}
\title{Isomorphism of graph classes related to the circular-ones property}
\author{Andrew R. Curtis\thanks{Cheriton School of Computer Science, University of Waterloo}
\and Min Chih Lin\thanks{CONICET, Instituto de C\'alculo and Departamento de Computaci\'on, FCEyN, Universidad de Buenos Aires}
\and Ross M. McConnell\thanks{Computer Science Department, Colorado State University}
\and Yahav Nussbaum\thanks{The Blavatnik School of Computer Science, Tel Aviv University}
\and Francisco J. Soulignac\thanks{CONICET, Departamento de Computaci\'on, FCEyN, Universidad de Buenos Aires}
\and Jeremy P. Spinrad\thanks{EECS Department, Vanderbilt University}
\and Jayme L. Szwarcfiter\thanks{Universidade Federal do Rio de Janeiro, Instituto de Matem\'atica}}

\date{}
\maketitle

\begin{abstract}
We give a linear-time algorithm that checks for isomorphism between two $0-1$ matrices that obey the circular-ones property.
This algorithm leads to linear-time isomorphism algorithms for related graph classes, including Helly circular-arc graphs,
$\Gamma$ circular-arc graphs, proper circular-arc graphs and convex-round graphs.
\end{abstract}

\section{Introduction}

Two graphs $G$ and $G'$ are \emph{isomorphic} if there is a bijection 
$\pi$ from the vertex set of $G$ to the vertex 
set of $G'$ that satisfies the following condition:
for all $u,v \in V(G)$, $u$ is adjacent to $v$ in $G$ if and only if 
$\pi(u)$ is adjacent to $\pi(v)$ in $G'$. The \emph{graph isomorphism problem} 
consists of determining whether two input graphs are isomorphic. The graph 
isomorphism problem is in $NP$, but it is neither known to be NP-complete
nor known to be in $P$.  However, the restriction of the problem to members
of graph classes with certain special topological properties is 
known to result in polynomial or even linear-time
algorithms.  The subject of this paper is the development of such algorithms
for a variety of such classes.

We show linear-time isomorphism algorithms for Helly circular-arc graphs,
$\Gamma$ circular-arc graphs, proper circular-arc graphs and convex-round graphs.
The common building block for all of these algorithms, is a linear-time
isomorphism algorithm for binary matrices that obey the circular-ones property,
which we show. 

In order to explain our results, we first establish some basic terminology.
A matrix is a {\em binary matrix} if all of its entries are 0 or 1.  
The \emph{adjacency matrix} of a simple graph $G$ is a binary square matrix $M$ 
that has $1$ in row $i$, column $j$, if vertex $i$ is adjacent to vertex $j$,
and 0 otherwise.  The 
\emph{augmented adjacency matrix} of $G$ is the adjacent matrix of the 
graph with $1$'s on the main diagonal.

A {\em clique matrix} of a graph is a binary matrix that has one row for
each vertex and one column for each maximal clique, and a 1 in row $i$, 
column $j$ if vertex $i$ is a member of maximal clique $j$.

A {\em consecutive-ones} matrix is a binary matrix whose columns 
can be ordered such that, in every row, the 1's are consecutive.
Such an ordering of the columns is a {\em consecutive-ones ordering}.
A {\em circular-ones} matrix is a binary matrix whose columns can be ordered
such that, in every row, either the 0's are consecutive or 1's are consecutive;
equivalently, the 1's are consecutive modulo the number of columns, and
the block of 1's is allowed to ``wrap around'' from the rightmost to the 
leftmost column.  Such an ordering of the columns is a {\em circular-ones
ordering}.  It is easily seen that the class of consecutive-ones matrices
is a proper subclass of the class of circular-ones matrices~\cite{Gol80}.

If a circular-ones ordering of a matrix with $n$ rows is known, then
it can be represented in $O(n)$ space by recording the number of
columns and listing, for each row, the
columns of the first and last 1 in the row.  
The column of the first 1 is the column
where a 1 follows a zero and the column of the last 1 is the one where
a 1 is followed by a 0.  (The column of the last 1 precedes
that of the first if the row wraps around the end of the matrix.)
A row that has only 0's or only 1's can be represented with
a suitable code, such as a single 0 or a single 1.
Let us call this a {\em succinct representation} of a circular-ones matrix.
There may be many succinct representations, since there may be
many circular-ones orderings of the matrix.

The {\em intersection graph} of a family of sets has one vertex for each
set in the family and an edge between two vertices if the corresponding
sets intersect.

A \emph{circular-arc graph} is the intersection graph of arcs on a 
circle. If we restrict circular-arc graphs to intersection graph of arcs on 
a circle such that no arc contains another, we get \emph{proper 
circular-arc graphs}~\cite{Gol80}. 

Another subclass of circular-arc graphs is the \emph{Helly circular-arc 
graphs}, sometimes called \emph{$\theta$ circular-arc graphs}.  
A circular-arc model has the {\em Helly property} if every family of
pairwise intersecting arcs has a common intersection point.
Such a model is called
a {\em Helly circular-arc model}.  A circular-arc 
graph is a Helly circular-arc graph if there exists a Helly circular-arc
model of it.
In~\cite{Gavril74}, it is shown that a graph is a Helly circular-arc graph if
and only if the clique matrix satisfies the circular-ones 
property.

The graphs whose augmented adjacency matrices satisfy 
the circular-ones property are \emph{$\Gamma$ circular-arc 
graphs}, also called \emph{concave-round graphs}. This graph class is also 
a subclass of circular-arc graphs~\cite{Tuck71}. 

The {\em complement} $\overline{G}$ of an undirected graph $G$
has the same vertex set, and an edge between two vertices if and only
if there is no edge between them in $G$.

The complement of a $\Gamma$ circular-arc graph has the circular-ones 
property for its adjacency matrix.   This class of graphs is called 
\emph{convex-round graphs}~\cite{BHY00}.

Wu~\cite{Wu83} presented the first polynomial algorithm for 
circular-arc graph
isomorphism, but later Eschen~\cite{EschenThesis} claimed to find a flaw in it. 
Hsu~\cite{Hsu95} presented an $O(nm)$ isomorphism algorithm for 
circular-arc graphs where $n$ denotes the number 
of vertices and $m$ denotes the number of edges in a graph.
In Section~\ref{sect:hsu} we give a counterexample to the
correctness of this algorithm.  We also describe
there a suggestion given by Hsu for a possible fix for the algorithm.
Therefore, there are currently no known efficient isomorphism algorithms 
for circular-arc graphs.
Some subclasses of circular-arc graphs do have efficient 
isomorphism algorithms. Interval graphs~\cite{LB79}, co-bipartite circular 
arc graphs~\cite{EschenThesis}, and proper circular-arc graphs~\cite{LSS08} all 
have linear-time isomorphism algorithms, while $\Gamma$ circular-arc graphs 
have an $O(n^2)$ isomorphism algorithm \cite{Che00}.

Two binary matrices $M_1$ and $M_2$ are {\em isomorphic} if there exists 
a permutation $\tau$ of the rows of $M_1$ and a permutation $\pi$
of its columns that makes $M_1$ identical to $M_2$.  If $\pi$
is known, $\tau$ is trivial to find by matching up identical rows
of $M_1$ and $M_2$.  Therefore, abusing notation somewhat, we
will sometimes call $\pi$ an {\em isomorphism} from $M_1$ to $M_2$, omitting
mention of $\tau$, and treat a matrix as a multiset of row vectors.

The results of the paper are organized as follows.  In 
Section~\ref{sect:preliminaries}, we give basic definitions and
review the {\em PC tree} of a circular-ones matrix~\cite{HsuMcC03,ShihHsuPlanarity}.
This gives a representation of all circular-ones orderings of
a circular-ones matrix.

In Section~\ref{sect:PCQuotients}, we present a notion of {\em quotient labels}
on the PC tree, which were developed in~\cite{Cur07}.
The matrix can be reconstructed from the tree
and its labels, establishing the quotient-labeled PC tree as a unique 
decomposition of a circular-ones matrix.

In Section~\ref{sect:matrixIsomorphism}, we 
give an algorithm that uses the quotient-labeled PC tree to test
isomorphism of circular-ones matrices that was also developed
in~\cite{Cur07}.  We define a notion of isomorphism
of quotient-labeled PC trees, show that two circular-ones matrices
are isomorphic if and only if their quotient-labeled PC trees are
isomorphic, and reduce the problem to testing whether the two matrices'
PC trees are isomorphic.  The running time is 
linear in the number of rows, columns and 1's of the matrix if
the circular-ones orderings are not provided, or linear in the
number of rows if a succinct representation of circular-ones matrices 
is provided.  

In Section~\ref{sect:Helly}, we reduce the problem
of testing whether two Helly circular-arc graphs are isomorphic
to testing whether two circular-ones matrices are isomorphic.
This gives a bound that is proportional to the number of vertices
and edges, or just proportional to the number of vertices, depending
on whether the graphs are represented with adjacency lists or with
suitable sets of circular arcs. A preliminary version of part of the
results of this section appeared in~\cite{LMSS08}.

In Section~\ref{sect:Gamma}, we use the fact that testing isomorphism of
$\Gamma$ circular-arc graphs and of convex-round graphs reduces to testing
isomorphism of circular-ones
matrices, giving an $O(n+m)$ or an $O(n)$ bound, depending on whether
a succinct representation is given. This leads to a new algorithm for testing
the isomorphism of proper circular arc graphs, which can run in $O(n)$ time
if two circular arc models are given.  

In Section~\ref{sect:hsu} we discuss the circular-arc isomorphism algorithm
of \cite{Hsu95}. We show a counterexample for this algorithm, and give a direction
suggested by Hsu for a possible fix.

\section{Preliminaries}\label{sect:preliminaries}

We consider simple undirected graphs $G$ and $G'$. We denote the number 
of vertices of $G$ by $n$ and the number of edges by $m$. We assume that 
$G'$ has the same number of vertices and edges as $G$, since otherwise it 
is trivial to see that the graphs are not isomorphic.

In this paper, we also let $n$ denote the number of rows of a binary matrix.
(Many of the matrices we deal with are derived from graphs, and have
one row for each vertex of the graph.)
The {\em size} of a binary matrix $M$, denoted $size(M)$,
is the number of rows plus the number of  columns plus the number of 1's; 
this is proportional to the number of words required to store the matrix 
using a standard sparse-matrix representation.  We will say that
an algorithm whose inputs are binary matrices runs in {\em linear time} 
only if it runs in time that is linear in this measure of the size of
the matrices.

By $N(v)$, we denote the set of neighbors of a vertex $v$.
By $N[v]$, we denote the {\em closed neighborhood} $\{v\} \cup N(v)$ of $v$.
If $U$ is a set of vertices then $N[U]$ is the union of $N[v]$ for 
all $v \in U$.

An (unrooted) tree $T$ is an undirected graph that is connected and has no
cycles.  
{\em Rooting} a tree $T$ at node $w$ consists of orienting all edges
so that they are directed from vertices that are farther from $w$ to
vertices that are closer to $w$, yielding a directed graph.  This gives
each vertex $u \neq w$ a unique outgoing edge $(u,v)$.  The neighbors
of $u$ can be classified as the (unique) {\em parent} of $u$ (the unique 
neighbor that is closer to $w$) and the {\em children} of $u$.
Once a tree has been rooted, we will continue to refer to the parent and 
children of $u$ as its neighbors.

A \emph{circular-arc model} of $G$ is a mapping from the vertices of $G$ 
to arcs on a circle such that two vertices are adjacent if and only if 
the corresponding arcs intersect. We represent a circular-arc model $\cal A$ by a 
cyclic doubly-linked list of the endpoints arcs.
Each vertex of $G$ has two endpoints in the model, one of them is a clockwise endpoint and the
other is a counterclockwise endpoint. 

A \emph{proper circular-arc model} is a circular-arc model in which no arc contains another. 
An \emph{interval model} is a circular-arc model whose circle has some 
point that is not contained in any arc.

Let $t = (t_1, t_2, \ldots, t_k)$ be a list, where each $t_i$ is a tuple
of integers.  
By {\em sorting a list of tuples}, we mean that we rearrange the
order of the $t_i$'s so that they are in increasing lexicographic order.
Let $L_1$ and $L_2$ be two lists of tuples.  Let $i$ be the first
position in which they differ.  Then we say that $L_1$ {\em precedes $L_2$
lexicographically} if the tuple at position $i$ of $L_1$ lexicographically
precedes the tuple at position $i$ of $L_2$, or else $L_1$ has length $i-1$.

\subsection{Bipartitive decomposition trees}

Let a {\em bipartitive tree $T$} on universe $V$ be an undirected tree 
such that its leaves are the elements of $V$, all internal nodes 
have degree at least three, and each internal node is labeled either
{\em prime} or {\em degenerate}.  
Let a {\em neighbor set} of a node $u$ denote the set of leaves in
a tree of the forest that results when $u$ is removed from $T$.
Since all internal nodes of $T$ have degree at least three, all members
of a neighbor set are leaves of $T$, so each neighbor set is a subset of $V$.
The {\em set family ${\cal F}(T)$ 
represented by $T$} consists of the following sets:

\begin{itemize}
\item A neighbor set of a prime node or the union of all but one of the neighbor
sets of a prime node;
\item Any union of at least one and fewer than all neighbor sets of a 
degenerate node.
\end{itemize}

Not all set families can be represented by a bipartitive tree.  Next,
we characterize those that can.

The {\em symmetric difference} $X \Delta Y$ of two sets $X$ and $Y$
is the set $(X \setminus Y) \cup (Y \setminus X)$.
Let us say that two subsets $X$ and $Y$ of universe $V$ {\em strongly overlap}
if $X \cap Y$, $X \setminus Y$, $Y \setminus X$ and $\overline{X \cup Y} =
V \setminus (X \cup Y)$ are all nonempty.

A {\em bipartitive set family on universe $V$} is a set family ${\cal F}$ with
the following properties:

\begin{itemize}
\item $\emptyset, V \not\in {\cal F}$

\item $\{x\} \in {\cal F}$ for all $x \in V$

\item $\overline{X} \in {\cal F}$ for all $X \in {\cal F}$

\item Whenever $X,Y \in {\cal F}$ strongly overlap, $X \cap Y$,
$X \cup Y$ and $X \Delta Y$ are all members of ${\cal F}$.
\end{itemize}

\begin{theorem}[\cite{CE80}]\label{thm:bipartitiveTree}
If $T$ is a bipartitive tree on universe $V$, then ${\cal F}(T)$ is a 
bipartitive set family on universe $V$.  Conversely, if ${\cal F}$
is a bipartitive set family, then there exists a unique bipartitive tree
$T$ such that ${\cal F} = {\cal F}(T)$.
\end{theorem}

\subsection{PC trees}

Let $M$ be a circular-ones matrix.  A row of $M$ can be thought of as the
bit-vector representation of a set, that is, it is
the set $X$ of columns of $M$ where the row has a 1.  Let $V$ denote
the columns of $M$ and let ${\cal R}$ denote the family of sets represented
by the rows.  Note that ${\cal R}$ is a set family on universe $V$.
Let ${\cal N}({\cal R})$ denote the family of subsets of $V$, excluding
$\emptyset$ and $V$ itself, that do not strongly overlap any member of
${\cal R}$.

\begin{lemma}[\cite{HsuMcC03}]
${\cal N}({\cal R})$ is a bipartitive set family on universe $V$.
\end{lemma}

It follows that ${\cal N}({\cal R})$ is represented by a bipartitive tree.
For historical reasons, the
prime nodes are known as {\em C nodes}, the degenerate nodes
are known as {\em P nodes} and the bipartitive tree is called a
{\em PC tree}~\cite{HsuMcC03,ShihHsuPlanarity}.  
Figure~\ref{fig:PCExample} gives an example.  When drawing a PC tree,
we use the convention
of representing a C node with a double circle and a P node with
a dot. In this figure, the neighbor sets of $c$
are $\{1,2,3\}$, $\{4\}$, $\{5\}$, $\{6\}$, and $\{7,8,9,10\}$.  
Each of these sets and the union of all but 
any one of these sets is a member of ${\cal N}({\cal R})$.  For example,
the neighbor set $\{7,8,9,10\}$ is a subset of rows 
$\{3,5,13,14\}$, it contains rows
$\{6,7,8\}$, it does not strongly overlap row $9$ because the union
of row 9 and $\{7,8,9,10\}$ is the entire column set, and it is disjoint
from all other rows.  Since it does not strongly overlap any row,
it is a member of ${\cal N}({\cal R})$.  Similarly, the union of
all neighbor sets other than $\{7,8,9,10\}$ (the complement of $\{7,8,9,10\}$),
is a member of ${\cal N}({\cal R})$.  However, because $c$ is a C node (prime)
the union of neighbor
set $\{6\}$ and $\{7,8,9,10\}$, which is neither the union of one
neighbor set nor the union of all but one, is not a member of 
${\cal N}({\cal R})$.  This is verified by observing that
it strongly overlaps row 12.

The neighbor sets of $a$ in the figure are $\{1\}$, $\{2\}$, $\{3\}$, $\{4,5, \ldots, 10\}$.
Because it is a P node (degenerate), every union of at least one and fewer 
than all of these sets is a member of ${\cal N}({\cal R})$, and this
is easily verified using similar checks.

Only unions of neighbor sets of a single internal node can be members
of ${\cal N}({\cal R})$.  For example, $\{1,2,10\}$ is not a union of
neighbor sets of a single node, and it is not a member of ${\cal N}({\cal R})$
because it strongly overlaps rows $2$, $4$, $8$, $10$, $11$, and $14$.

Note that ${\cal R}$ has the circular-ones property, and
the figure depicts a way to cyclically order the edges incident to
each node such that in the resulting tree, every member of ${\cal R}$
is consecutive in the circular ordering of leaves.  This is an
example of a general phenomenon:

\begin{theorem}[\cite{HsuMcC03}]
Let $T$ be the bipartitive tree for ${\cal N}({\cal R})$, where ${\cal R}$
is the set family represented by rows of a circular-ones matrix with column
set $C$.  Then:

\begin{itemize}
\item The edges incident to internal nodes can be cyclically ordered
in such a way that the resulting cyclic order of leaves is a circular-ones
ordering of the matrix;

\item Reversing the cyclic order of edges about a prime node imposes a
new circular-ones ordering on the leaves.

\item Arbitrarily permuting the cyclic order of edges about a degenerate node
imposes a new circular-ones ordering on the leaves.

\item All circular-ones orderings of the leaves are obtainable by
a sequence of these two operations.
\end{itemize}
\end{theorem}

This gives a convenient data structure, the {\em PC tree}, for
representing all circular-ones orderings of a circular-ones matrix.

For example, in Figure~\ref{fig:PCExample}, permuting the counterclockwise
cyclic order of neighbors about $a$ so that it is $(1,c,2,3)$ and reversing
the cyclic order of neighbors about $b$ so that it is
$(c,10,9,8,7)$ imposes a new cyclic leaf order on the tree,
$(1,4,5,6,10,9,8,7,2,3)$, which is easily seen to be a circular-ones
ordering.

\begin{figure}
\centerline{\includegraphics[]{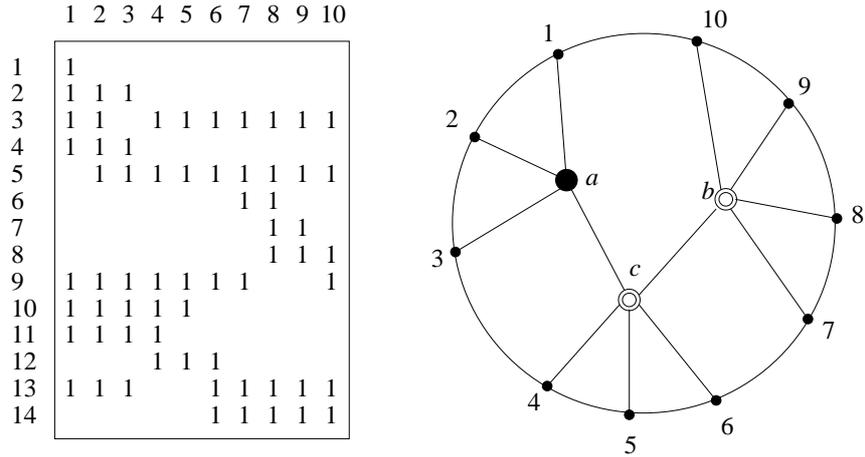}}
\caption{The PC tree for a circular-ones matrix.  Nodes $b$ and $c$
are C nodes (double circles) and node $a$ is a P node.}\label{fig:PCExample}
\end{figure}

To gain an insight into why this works, consider a circular-ones ordering
of columns of a circular-ones matrix, let $X$ be a circularly consecutive
set of columns, and let $R$ be the set represented by some row.  If $X$
is removed and reinserted in reverse order, it will disrupt the consecutiveness
of $R$ if and only if it strongly overlaps $R$.  Since ${\cal N}({\cal R})$
is the family of sets that don't strongly overlap any column, they are
the sets that can be reversed in the cyclic order without disrupting
the circular consecutiveness of any row.  Each allowed rearrangement of
the PC tree corresponds to a sequence of such reversals, and each such
reversal is allowed by the PC tree.  

The PC tree was first described by Shih 
and Hsu~\cite{ShihHsuPlanarity}.  Its relationship to bipartitive set families and
the circular-ones orderings of matrices was first described in~\cite{HsuMcC03}.

It takes time linear in the size of a circular-ones matrix
to build the PC tree even
when a circular-ones ordering is not given as part of the 
input~\cite{HsuMcC03,ShihHsuPlanarity}.
As part of the output, a circular ordering of the edges incident to
each internal node is given, which imposes a circular-ones ordering on
the leaves.  This gives a representation of all circular-ones orderings
of the matrix in linear time.

Henceforth, we will assume for simplicity that every row 
of a circular-ones matrix has at least one 1 and one 0; any row without
this property is irrelevant
to the circular-ones arrangements.  

\section{Quotient labels for the PC tree}\label{sect:PCQuotients}

In this section, we give a way to label the nodes of the PC tree of
a circular-ones matrix with {\em quotients} so that the matrix can be 
reconstructed from the labeling.  These results were developed 
in~\cite{Cur07}, and the scheme is similar to ones
developed for PQ trees in~\cite{KorteMoehringPQ,LB79,McCSODA04}.

Recall that we assume that every row of the 
circular-ones matrix $M$ has at least one 1 and at least one 0, 
since a row that is all zeros or all ones has no
effect on the circular-ones orderings of the matrix.  Suppose
that $M$ has a circular-ones ordering of columns.  
Consider the ordered tree that results from the PC tree of $M$
by the cyclic
ordering about each node that is forced by the cyclic leaf order
corresponding to the circular-ones ordering.
For each row $R$
that has at least two 1's we perform the following procedure.
Root the PC tree at a leaf corresponding to a column with a 0 in $R$.
We maintain the order about each node $u$ in the following way. If
 $(w_1, w_2, \ldots, w_k)$ is the counterclockwise order about $u$
and $w_i$ becomes the parent of $u$, then $(w_{i+1}, w_{i+2}, \ldots,
w_k, w_1, \ldots, w_{i-1})$ becomes the linear order of children
in the resulting ordered {\em rooted} tree.
Let $u$ be the least common ancestor of the 1's in $R$.  Let $X$
be the leaf descendants of a child of $u$.  Since  $X \in {\cal N}({\cal R})$,
$X$ does not strongly overlap $R$.  
Therefore, $X$ is either a subset
of $R$ or disjoint from $R$.  It follows that $R$ consists of the leaf
descendants of two or more children of $u$.   Moreover, since the cyclic
order of edges about $u$ gives the circular-ones orderings of $M$, $R$ consists
of the neighbor sets of $u$ through a consecutive set $R'$ of children of $u$.
Let $R'$ be the {\em projection} of $R$ on $u$.    
Figure~\ref{fig:projection} illustrates the concept.

\begin{figure}
\centerline{\includegraphics[]{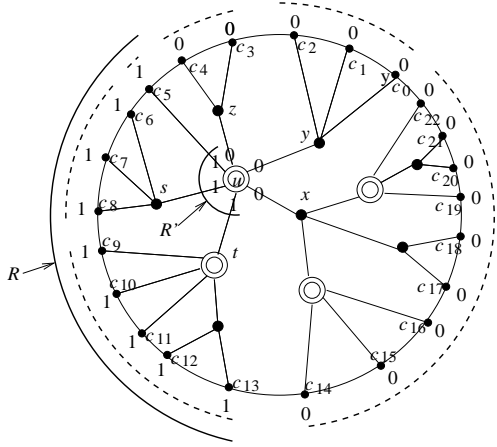}}
\caption{The projection $R'$ of a row $R$ of a circular-ones matrix.  
Each element of the row corresponds to a leaf of the PC tree.  Rooting 
the tree at any leaf where the row has a 0 and then finding the least 
common ancestor of the 1's gives the node $u$ that the row projects to.  
The projection makes up a row vector in the quotient matrix at $u$
that has one column for each neighbor of $u$, ordered in cyclic order.
Because every neighbor set is a member of ${\cal N}({\cal R})$,
the columns where a row has 1's will always be a union of neighbor sets 
of $u$.   In this case, $R$ has 1's in $\{c_5, \ldots, c_{13}\}$, and
the neighbor sets of $u$ (dashed arcs) are $\{c_0, c_1, c_2\}$, 
$\{c_3,c_4\}$, $\{c_5\}$, $\{c_6, c_7, c_8\}$, 
$\{c_{9}, c_{10}, \ldots, c_{13}\}$, and $\{c_{14}, c_{15}, 
\ldots, c_{22}\}$.  The ones that are subsets of $R$ are reachable
through neighbors $C_5, s, t$.  The projection $R'$ of $R$ is a row
of $u$'s quotient matrix that has 1's in columns $C_5, s, t$, and
0's in columns $x,y,z$.}\label{fig:projection}
\end{figure}

As a special case, if $R$ has only one 1, let $c$ be the column where
the 1 occurs, and let $u$ be its neighbor.
We consider the projection $R'$ of $R$ to be on $u$, and to consist
of $u$'s neighbor set $\{c\}$.  

When this projection has been performed on all rows, 
each internal node $u$ has received
the projection of zero or more rows of $M$, and we represent each projection
with a {\em row vector} in a {\em quotient matrix} whose columns are neighbors
of $u$.  The row has a 1 in column $w$ if the leaves reachable through $w$
are 1's of $R$ and a 0 if they are 0's in $R$.
The result is a matrix whose rows are sets of neighbors of $u$.
Note that the quotient at a node may be empty.  

Figure~\ref{fig:PCQuotients} illustrates the quotients for the example
of Figure~\ref{fig:PCExample}.  (The tree has been rooted at node $c$;
the motivation for this is explained below.)

\begin{figure}
\centerline{\includegraphics[]{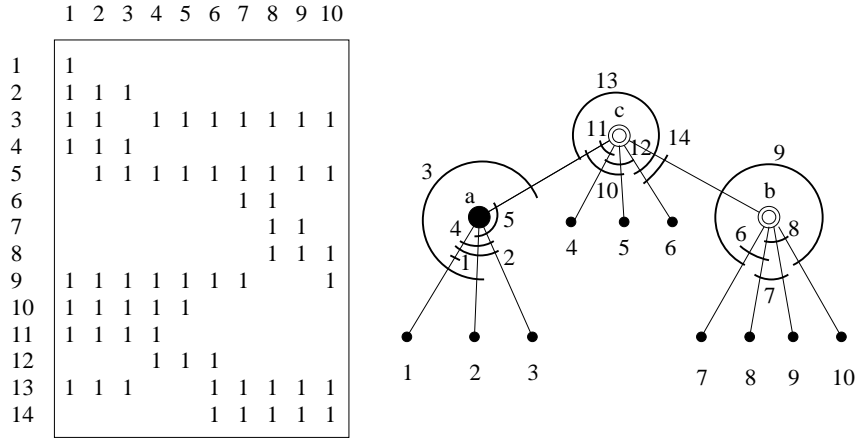}}
\caption{The quotients at the 
internal nodes of Figure~\ref{fig:PCExample}.
The tree has been rooted at its center.
Rows $\{1,2, \ldots, 5\}$ project to node $a$, rows $\{6,7,8,9\}$
project to node $b$, and rows $\{10, 11, \ldots, 14\}$ project
to node $c$.}\label{fig:PCQuotients}
\end{figure}

\begin{lemma}\label{lem:PQuotients}
At a P node $u$ with $k$ neighbors, each row of the quotient matrix consists of
either $k-1$ 1's and one 0 or $k-1$ 0's and one 1.
\end{lemma}

\begin{proof}
Let $M$ be a circular-ones matrix, and
let ${\cal R}$ be the family of sets of columns obtained by considering
each row to be the bit-vector representation of a set.

By the definition of a P node, every union of 
neighbor sets of $u$ is a member of 
${\cal N}({\cal R})$.  
Let $R'$ be the projection of some row $R$ to the quotient
at $u$.  Suppose $R'$ consists of more than one and fewer than $k-1$
neighbors of $u$.  Then $R$ consists of more than one and fewer than
$k-1$ neighbor sets of $u$.
Let $x$ and $y$ be two neighbors of $u$, one in $R'$
and one not.  Let $X$ and $Y$ be the neighbor sets of $u$ reachable
through $x$ and $y$.  Then $X \cup Y$ is a union of neighbor sets
of $u$ that strongly overlaps $R$.  This disqualifies it as a member
of ${\cal N}({\cal R})$, contradicting the fact that every union
of at least one and fewer than all neighbor sets of a P node is a 
member of ${\cal N}({\cal R})$.
\end{proof}

The quotient at node $a$ of Figure~\ref{fig:PCQuotients} 
illustrates the phenomenon.  The projections of rows 2 and 4
exclude only neighbor $c$, the projection of row 5 excludes only neighbor
1, the projection of row 3 excludes only neighbor 3, and
the projection of row 1 contains only neighbor 1.

\begin{lemma}
Any cyclic ordering of the quotient at a P node is a circular-ones ordering
of it.  The quotient at a C node has a unique circular-ones ordering,
up to rotation and reversal.
\end{lemma}

\begin{proof}
Again, let $M$ be a circular-ones matrix, and
let ${\cal R}$ be the family of sets of columns obtained by considering
each row to be the bit-vector representation of a set.

The result follows immediately for P nodes by Lemma~\ref{lem:PQuotients}.
Suppose $u$ is a Q node and let $k$ be the number of neighbors
of $u$.  Let $X$ be an arbitrary union of more than one and fewer than $k-1$
neighbor sets of $u$.  By the definition of a Q node, 
$X \not\in {\cal N}({\cal R})$.  Therefore, it strongly overlaps
some row $R$ of $M$.  Let $y$ be a column not in $X \cup R$.
Root the PC tree at $y$.  If $R$ projects to a proper descendant of $u$,
then $R$ is either a subset of $X$ or disjoint from $X$, a contradiction
to the strong overlap.
If $R$ projects to a proper ancestor of $u$, then $X \subset R$, also a 
contradiction.  If $R$ projects to a node that is neither an ancestor
nor a descendant of $u$, then $R$ is disjoint from $X$, once again
a contradiction.  Therefore, $R$ projects on $u$, and the projections
$X'$ and $R'$ of $X$ and $R$ on $u$ are strongly overlapping sets of 
neighbors of $u$.

We conclude that every set $X$ of more than one and fewer than $k-1$
neighbors of $u$ strongly overlaps some row $R$ of the quotient at $u$.
Reversing the order of members of $X$ in the cyclic ordering of
neighbors about $u$ disrupts consecutiveness of $R$.  This implies that
the PC tree of the quotient matrix at $u$ has a single internal node,
a C node.  The circular-ones orderings of this quotient are unique up
to rotation and reversal.
\end{proof}

As an illustration, it is easily verified
in Figure~\ref{fig:PCQuotients} that every ordering
of columns of the quotients at P nodes $a$ is a circular-ones ordering,
and that the quotient at C nodes $b$ and $c$ each have only two circular-ones 
orderings, one of which is depicted and the other of which
is its cyclic reversal.

\subsection{Computing the quotient-labeled PC tree in time linear in $size(M)$}

Note that the linear-time PC-tree construction algorithm of~\cite{HsuMcC03}
gives a circular-ones arrangement of the leaves, but does not label
the nodes with their quotient.  
Let $u$ be the least common ancestor of a row $R$ when the tree
is rooted at some leaf not in $R$.
Let a node be
{\em black} if all of its leaf descendants are in $R$.  Because
all neighbor sets of $u$ are either subsets of $R$ or disjoint from it,
the black nodes consist of all ancestors of the leaves in $R$, up to,
and possibly including $u$, if all of its children are black.
The black children of $u$ are the ones corresponding
to the projection of $R$.

To find the projection of $R$, we
can therefore blacken the leaves that are members of $R$.  When we blacken
a node we increment a counter in the parent, so that it keeps track of
how many blackened children it has.  When a node's counter is incremented
to a value that equals its number of children, we blacken it, and increment
its parent's counter.  It is easily seen by induction on the height of
a node that a node is blackened by this procedure if and only if it is black.

The procedure halts when no new vertex can be marked.
The blackened nodes induce a rooted forest, which is 
a rooted tree if $u$ is black.  Every internal node of this forest
or tree has at least two black children.  Therefore, the number of black
internal nodes is bounded by the number of black leaves, which is the size of $R$.
Since the procedure spends $O(1)$ time whenever it blackens a node,
it takes $O(|R|)$ time to find the least common ancestor $u$
and the edges to neighbor sets that make up $R$'s projection.

The procedure
can be simulated without actually rooting the tree at $y$; it suffices to
work ``inward'' from the marked leaves that are members of $R$, and mark
a node when its counter is equal to one less than its number of neighbors,
since one of its neighbors is implicitly the parent.  Since the entry in
the quotient matrix has a 1 for each member of $X$, it takes $O(|R|)$ time
to find the quotient representation of each row $R$.  The total time
to find all quotients is therefore linear.

\subsection{Computing the quotient-labeled PC tree in $O(n)$ time}

We now show that given a succinct representation of a circular-ones
matrix, we may obtain the PC tree in $O(n)$ time.
By {\em complementing} a row of a binary matrix, we mean changing every 1 in 
the row to a 0 and every 0 to a 1.  It is obvious that complementing
a row of a circular-ones matrix yields a circular-ones matrix~\cite{Gol80}.
The following is given in~\cite{HsuMcC03} for the version of the PC tree where
the nodes do not have quotient labels.  

\begin{lemma}
The PC tree of a circular-ones matrix is invariant under complementing a
row of the matrix.
\end{lemma}

The lemma is obvious from the observation that ${\cal N}({\cal R})$ does not
change when a row is complemented.

Let $c$ be a column of a circular-ones matrix.  Complementing all of the
rows that have a 1 in column $c$ turns $c$ into a column of zeros.  
In the succinct representation of the matrix, it takes $O(1)$ time to
complement each such row: if $(f,l)$ represents the columns of
the first and last 1 in the block of 1's in the row,
then the $(l+1,f-1)$, modulo the number of columns, 
represents the complement of the row.
Since now no row contains $c$, a circular-ones ordering where $c$ is
the last row of the matrix is a {\em consecutive-ones} ordering on the columns
of the matrix, excluding $c$.

The {\em PQ tree} was developed by Booth and Lueker to represent all
consecutive-ones orderings of the columns of a consecutive-ones 
matrix~\cite{BL76}.  It is a certain rooting
of the PC tree.  (The PC tree was originally developed to provide
an easier way to compute the PQ tree.~\cite{ShihHsuPlanarity})

\begin{lemma}[\cite{HsuMcC03}]\label{lem:PCPQ}
Let $M$ be a circular-ones matrix and $c$ be a column.  Let $M'$ be
the result of complementing all rows of $M$ that have a 1 in column $c$,
and then removing column $c$.
Removing leaf $c$ from the PC tree of $M$ and then rooting it at the 
neighbor of $C$ gives the  PQ tree for $M'$.
\end{lemma}

\begin{lemma}\label{lem:PCOn}
It takes $O(n)$ time to compute the PC tree of a circular-ones matrix,
given a succinct representation of it.
\end{lemma}
\begin{proof}
Given a succinct representation of a circular-ones matrix $M$, we
may remove all rows that have a single 1 or a single 0, since these
have no effect on the circular-ones orderings of the matrix, hence on
the PC tree.
We may then select the last column $c$ of $M$, and identify the rows with
a 1 in column $c$ 
in $O(n)$ time.  We may complement each of them in $O(1)$ time, as
described above, for a total of $O(n)$ time.
We may remove $c$ from the resulting matrix by decrementing the
record of the number of columns.
This takes $O(n)$ time and gives a succinct representation of
a consecutive-ones matrix, $M'$.

In~\cite{McCPQAlg}, an $O(n)$ algorithm is given for finding the PQ tree
of a consecutive-ones matrix, given a succinct representation.
By Lemma~\ref{lem:PCPQ} and the fact that the removed rows, which have all
0's or all 1's, have no effect on the PQ tree of $M'$,
we may then add a new leaf corresponding
to $c$ to the root of the tree produced by this algorithm, and then unroot it 
to obtain the PC tree of $M$, also in $O(n)$ time.
\end{proof}

Note that since the quotient labels of a PC tree have circular-ones
orderings of columns, they can be expressed in the form
of succinct representations also.  Doing this for all quotient labels
causes them to take $O(n)$ space, since there is one quotient row for
each row of the original matrix.  This raises the question of whether
we can find this {\em succinct quotient-labeled PC tree} in $O(n)$
time.

\begin{lemma}
Given a succinct representation of a circular-ones matrix with $n$ rows,
it takes $O(n)$ time to find the succinct quotient-labeled PC tree.
\end{lemma}
\begin{proof}
The algorithm proceeds as in the proof of Lemma~\ref{lem:PCOn} until 
the PQ tree of the succinct representation of the consecutive-ones
matrix $M'$ has been computed.  We install the succinct quotient labels in 
this PQ tree, as follows.  
We use Harel and Tarjan's least-common ancestor algorithm~\cite{HT84} to
find the least common ancestor of the endpoints of each interval,
and the child that contains the right endpoint of each interval.
Reversing the tree and repeating this gives the child containing the
left endpoint.   This takes $O(n)$ time in total.  These become the beginning 
and ending points for the representation of the row in the quotient at its
least common ancestor.

For every row that was complemented to obtain $M'$ from $M$, we
complement the image of the row
in the quotient at its least common ancestor, which takes 
$O(1)$ time for each of these rows.  Note that if the image is
all but one neighbor $w$ of a node $u$ and $w$ is not a leaf, then,
after the complementation, this projection consists only of $w$.
However, the definition of the projection at the beginning of Section 3
dictates that in this case the projection be onto the least common ancestor
of the row when the tree is rooted at a node that is 0 in the row.
Now $w$, not $u$, is this node, so the
the projection must be moved to $w$, and be changed to consist of all 
neighbors of $w$ other 
than $u$.  

For each row that was removed because it has only one 1 or one 0, let
$d$ be the leaf of the PC tree (column of the matrix) where the 1 or the
0 occurs.  The projection is on the neighbor $w$ of $d$, and consists
either of $d$ or of all neighbors of $w$ other than $d$, depending on
whether the row had a 1 or a 0 in $d$.
\end{proof}

\section{Testing isomorphism of circular-ones matrices}\label{sect:matrixIsomorphism}

In this section, we give an $O(size(M))$ time algorithm, first developed
in~\cite{Cur07}, to test isomorphism
of two circular-ones matrices if no circular-ones ordering of the two
matrices is given, and an $O(n)$ algorithm if a succinct representation
is given.  

It is possible to test whether the matrices are both circular-ones matrices
in linear time and to find a circular-ones ordering if they 
are~\cite{Tuck71}.
If both of them fail to be circular-ones matrices, the input is rejected
for failing to meet the precondition.  If exactly one of them fails to
be a circular-ones matrix, they are not isomorphic.  Otherwise,
from this ordering, we may find a succinct representation of the
two matrices in linear time.  We may then produce the quotient-labeled
PC trees with a succinct representation of the quotients in $O(n)$ time.

If the numbers of rows that are all 0's in the two matrices differ
or the number of rows of all ones differ, the matrices are not isomorphic.  
Otherwise, the problem reduces to the question of whether the two matrices are
isomorphic when these rows are eliminated.  Let $M_1$ and $M_2$ be
these two matrices.  This allows us to continue under our assumption that
every row of the matrices has at least one 0 and at least one 1.

Let us define an {\em isomorphism} $\pi$ from
one {\em quotient-labeled} PC tree, $T$, to another one, $T'$ to 
consist of the following.  It must be that whenever $u$ and $v$ are nodes 
of $T$, then $\pi(u)$ and $\pi(v)$ are adjacent in $T'$ if and only if 
$u$ and $v$ are adjacent in $T$.  This is the standard notion of graph 
isomorphism.  It must also satisfy an additional constraint.
Each neighbor $w$ of internal node $u$
of $T$ corresponds to a column of the quotient matrix at $u$.
Each neighbor $w'$ of node $u'$ of $T'$ is a column of the quotient
matrix at $u'$.  If $\pi$ is an isomorphism of the underlying trees,
then let $\pi_u$ denote the bijection it induces from neighbors of $u$
to neighbors of $u'$.  It must be  the case that at each internal
node $u$, $\pi_u$ is an isomorphism from the quotient matrix at $u$
to the quotient matrix at $u'$.  If such a $\pi$ exists,
we say that the trees are {\em isomorphic as 
quotient-labeled PC trees}; otherwise we say that they are not.
This definition precludes mapping a P node to a Q node, because
the quotient label determines whether a node is a P node or a Q node.

\begin{lemma}\label{lem:PCIsomorphism}
Two circular-ones matrices are isomorphic if and only if their quotient-labeled
PC trees are isomorphic.
\end{lemma}
\begin{proof}
Suppose matrices $M$ and $M'$ are circular-ones orderings of isomorphic 
circular-ones matrices.  Let $\pi$ be an isomorphism from $M$ to $M'$.
After permutation of columns of $M$ by $\pi$, $M$ and $M'$ have
identical quotient-labeled PC trees, since they are identical multisets of 
row vectors. 
Since $M$ and $M'$ are both circular-ones orderings of $M$,
$\pi$ is one of the permutations of columns of $M$ permitted
by the quotient-labeled PC tree of $M$.  In other words, the PC
tree of $M'$ can be obtained from the PC tree of $M$ by the
permitted operations.  These operations define an isomorphism from
the quotient-labeled PC tree of $M$ to the quotient-labeled PC tree of $M'$.

Conversely, suppose $M$ and $M'$ are circular-ones orderings of matrices that
have isomorphic quotient-labeled PC trees.  Use the isomorphism 
to arrange the PC tree of $M$ so that it is identical to that of $M'$.
Since it now represents $M'$ instead of $M$, it follows that the
permutation of leaves of the PC tree induced by the isomorphism
is an isomorphism from $M$ to $M'$.
\end{proof}

Lemma~\ref{lem:PCIsomorphism} is the basis of our algorithm:  for two 
circular-ones matrices, we compute their quotient-labeled PC trees and
test whether they are isomorphic.  To test whether they are isomorphic,
we encode the trees with strings in such a way that they both have
the same encoding if and only if they are isomorphic.

\subsection{Testing isomorphism of rooted, unordered trees}\label{sect:AHU74}

Our starting point is an algorithm for testing isomorphism of rooted,
unordered trees that is given in the textbook~\cite{AHU74}.
A rooted tree is {\em unordered} if the left-to-right order of children of 
its nodes is not specified.
Two unordered, rooted trees $T_1$ and $T_2$ are isomorphic if they are 
isomorphic as 
directed graphs when the edges are oriented from child to parent.

If the trees do not have the same height, they are not isomorphic.  Otherwise,
the algorithm proceeds by induction by level, from level 0, which is the level
of the deepest leaf, to level $h$, which is the level of the root and the 
height of the trees.  At each level, it labels each node $u$ with an integer
{\em isomorphism-class label} $e_u$ such that two nodes at the level have the
same label if and only if the subtrees rooted at them are isomorphic.
At step $i$, we may assume by induction that this has been done for nodes
at level $i-1$.   For all leaves at level $i$, the isomorphism-class label
is 0.  For the remaining nodes, we may apply the following procedure.
Let $t_u$ be the tuple of isomorphism-class labels for the children of
$u$, sorted in nondecreasing order.
We may then sort the non-leaf nodes at level $i$
by lexicographic order of the tuples assigned to them, in order to group
identical tuples together.  If there are $k$ distinct tuples, among the
tuples at level $i$,
we then assign isomorphism labels from 1 through $k$ to the tuples,
and then each node with the number of the tuple it was assigned.  
By induction, this meets the precondition for
induction step $i+1$.

Therefore, two trees are isomorphic if their roots receive the same 
isomorphism class label.  

\subsection{Canonical encodings of quotients}\label{sect:canon}

An additional requirement of an isomorphism $\pi$ on 
quotient-labeled PC trees
is that when a node $u$ of $T$ maps to a node $\pi(u)$ in $T'$,
then the neighbors of $u$ map to neighbors of $\pi(u)$ in a way that
is an isomorphism of the quotient at $u$ to the quotient at $\pi(u)$.

This requires that $u$ and $\pi(u)$ have isomorphic quotients.
In this subsection, we give an encoding of the quotient at a node of a
quotient-labeled PC tree so that two nodes receive the same encoding
if and only if their quotients are isomorphic.

\subsubsection{P nodes}

Consider the P node $a$ of Figure~\ref{fig:PCQuotients}.  
An obstacle to an immediate canonical representation is that the quotient 
can be presented in any column order, since all column orders are circular-ones
orderings.  

By
Lemma~\ref{lem:PQuotients}, every row of the quotient consists either
of one of $u$'s neighbors, or of all but one of its neighbors.  We can
encode the quotient by giving, for each neighbor, an ordered pair of integers.
The first integer is the number of rows
that exclude only that neighbor and the second integer is number of rows 
that contain only that neighbor.  

For example, at node $a$ in Figure~\ref{fig:PCQuotients},
the tuple generated by neighbor 1 is $(1,1)$, since the projection
of one row, row 5, excludes only neighbor 1, and the projection
of one row, row 1, contains only neighbor 1.
The tuple generated at neighbor 2 is $(0,0)$ since no row
contains only 2 and no row excludes only 2.  The tuple generated
at neighbor 3 is $(1,0)$, since the projection of one row, row 3, excludes only 
neighbor 3 and no rows contain only neighbor 3.  The tuple generated at 
neighbor $c$ is $(2,0)$, since the 
projections of rows 2 and 4 exclude only $c$ and no row contains only $c$.
Note that the tuple generated at a neighbor is invariant under permutation
of the cyclic order of neighbors about the P node.

Listing these tuples in the cyclic order of the nodes at which they
are generated yields an encoding of the quotient.  
The next step of the construction is to sort the generated tuples
lexicographically to obtain a tuple of tuples.
In the example of node $a$ of Figure~\ref{fig:PCQuotients},
this gives $((0,0),(1,0),(1,1),(2,0))$.  
This encoding is unique, since it is the lexicographic
minimum of all tuples whose elements are the generated tuples.

Finally, in order to prevent the possibility that a P node and a C node
will have the same encoding, we prepend a reserved P-node flag, 0, to
the list of tuples.  In the example if Figure~\ref{fig:PCQuotients}, this
gives $(0,(0,0),(1,0),(1,1),(2,0))$.

This therefore gives a test of isomorphism of quotients at two P nodes
by generating their canonical representations and testing whether they are
equal.

\subsubsection{A canonical encoding of the quotient at a C node}

In generating a canonical encoding of the quotient at a C node $u$, there are
only two cyclic orders of the columns that are circular-ones orderings.
One is the reverse of the other.  This leaves us with two obstacles
to a canonical representation:  which of these two cyclic orders should 
we choose,
and where in the cyclic ordering should we begin in developing a tuple to
represent the quotient?

We begin by traveling
counterclockwise around the cycle, starting at an arbitrary point.
At each neighbor $w$, we generate a tuple that lists the lengths of rows
whose clockwise-most 1 is at $w$, and list them in nondecreasing
order.

For example, consider the C-node $c$ depicted in Figure~\ref{fig:PCQuotients}.  
At neighbor $a$, we see that there are two projections, of
rows 11 and 10, whose
clockwise endpoint is at $a$, and they have lengths 2 and 3, respectively.
Therefore, the tuple generated for neighbor $a$ is $(2,3)$.
At neighbor 4, there is one projection, of row 12, that has its
clockwise endpoint at $4$, and it has length 3, so the tuple generated for 
4 is $(3)$.  Similarly, the tuples generated for neighbors 5, 6, and $b$
are $()$, $(2,3)$, and $()$, respectively.

Assembling these tuples in clockwise order, we get 
$((2,3),(3),(),(2,3),())$.
However, we must consider that
we made an arbitrary decision in choosing the point on the circle at which
to start generating the tuples.  The effect of different choices is to
rotate the resulting tuple of tuples.  To choose it in a canonical way, we
choose the rotation of the generated list of tuples that is earliest
lexicographically:  $((),(2,3),(),(2,3),(3))$.
We also made an arbitrary decision in going around the cycle clockwise
instead of counterclockwise.  Therefore, we
repeat the above procedure counterclockwise, generating 
$((2),(3),(3),(2),(3))$ if starting at $b$, and then choose the
rotation of this that is lexicographically minimum, $((2),(3),(2),(3),(3))$.
To choose the direction of
travel in a canonical way, we choose, from these two lists of tuples,
the one that is earlier in lexicographic order:  $((),(2,3),(),(2,3),(3))$.

To avoid any possibility that a C node and a P node can get the same encoding,
we prepend a reserved C-node flag, 1, to the list, yielding
$(1,(),(2),(2,3),(),(3))$.

The general algorithm is as follows. First, we order the neighbors of the node $u$ so that the
quotient has the circular-ones ordering.  For each neighbor $w$ in 
counterclockwise order, list the lengths of the rows of the quotient whose
clockwise-most 1 is at $w$, in nondecreasing order.  This gives a tuple
$(l_1,l_2, \ldots l_k)$ for $w$.
The sequence of tuples generated for each neighbor $w$, taken in
counterclockwise order, gives a tuple of tuples.
Rotate this ordering to get the lexicographic minimum through
rotation of the tuples.  Then repeat
the exercise reversing the roles of clockwise and counterclockwise to
obtain another such set of tuples.  From these, select the one that
is earlier lexicographically.  Then prepend the reserved C-node flag 1
to this list.

If $\pi$ is an isomorphism, then the tuples generated at a neighbor
$w$ of $u$ is the same as the one generated at neighbor $\pi(w)$ of $\pi(u)$.
Since the mapping of the cyclic ordering neighbors of $u$ to neighbors 
of $\pi(u)$ is the cyclic ordering of neighbors of $\pi(u)$ or its reverse,
the quotient at $u$ and $\pi(u)$ are encoded by the same tuple.
Conversely, if the quotient at C nodes $u$ and $w$ have the same tuple,
then since each tuple uniquely encodes cyclic rotations of the quotient,
$u$ and $w$ have isomorphic quotients.  Therefore, isomorphism of two
quotients can be tested by determining whether they are encoded
by the same quotient.

\subsection{Testing isomorphism of quotient-labeled PC trees}\label{sect:PCiso}

We combine elements of the rooted-tree isomorphism test of 
Section~\ref{sect:AHU74} with the test for isomorphism of quotients
of Section~\ref{sect:canon}, in order to obtain an isomorphism test
for quotient-labeled PC trees.

The use of elements of Section~\ref{sect:AHU74} requires us to root the
PC trees.  Conceptually, a rooting of a tree may be viewed as an 
orientation of its edges from child to parent, yielding a directed graph.
An isomorphism $\pi$ from
one rooted tree, $T$, to another, $T'$, is a bijection from nodes of $T$
to nodes of $T'$ such that for $u,v \in V(T)$, $(\pi(u),\pi(v))$ is a directed
edge in $T'$ if and only if $(u,v)$ is a directed edge in $T$.
Once we root two PC trees, we require them to satisfy this 
condition.  We must therefore be careful to root the two trees
in an isomorphic way whenever they are isomorphic.

The {\em center} of a one-vertex tree is its vertex and the center
of a one-edge tree is its edge.  Otherwise, the center is obtained
by deleting its leaves and recursively finding the center of the resulting
subtree.  It consists of a single vertex or a single edge.

If the center of a PC tree is node, we root it at that node.
If it is an edge $vw$, we subdivide the 
edge with a pseudo-node $x$ and root it at $x$ so that the tree has a node 
as a root, as we did above.  In the quotient at $v$, replace the name of 
$w$ with $x$ and in the quotient at $w$, replace the name of $v$ with $x$.  
Now $x$ can be treated as a P node with an empty quotient.

Once we have rooted the trees, we proceed by induction on the 
level $i$, as in the algorithm of Section~\ref{sect:AHU74}.  Because
we are applying a stronger notion of isomorphism, which observes
constraints imposed by the quotients, we must redefine what is
meant by the isomorphism-class label $e_u$ assigned to a node $u$ at
level $i$.

For every node $u$ in $T$ we define a tree $T_u$, an induced subtree of $T$, as
follows.
If $u$ is a leaf node in $T$, let $T_u$ be the one-node tree 
consisting of $u$.  If $u$ is an internal node that is not the root of
$T$, let $T_u$ be the
quotient-labeled subtree induced by
$u$, its descendants, and the parent $w$ of $u$. 
If $u$ is an internal node and the root of 
$T$, then $T_u = T$, and $u$ is the root of $T_u$.
Two nodes $u$ and $u'$ at level $i$ are in the {\em same isomorphism
class at level $i$} if and only if $T_u$ and $T_{u'}$ are 
isomorphic as quotient-labeled rooted PC trees.

Because of the
inclusion of the parent of $u$, the neighbors
of each internal node of $T_u$ are the same in $T_u$ as they are in $T$.
This allows each internal node to retain the same quotient in $T_u$
as it has in $T$.  This tree is rooted
at the parent of $u$.  

In order to merge the constraints of Sections~\ref{sect:AHU74}
and~\ref{sect:canon}, we prepend the isomorphism class label $e_w$
of a node $w$ at level $i-1$ (Section~\ref{sect:AHU74}) to the tuple
generated for the node in the encoding of the quotient 
(Section~\ref{sect:canon}).

The preconditions at the beginning of the induction step at level $i$ are 
the following.  Leaves at level $i$ are labeled with equivalence-class label 
0.  If $u$ is a non-leaf at level $i$, then the parent $p$ of $u$, if it 
exists, is a node of $T_u$ and labeled with equivalence class label 
-1.  Note that $p$ is a leaf of $T_u$, but no automorphism
$T$ to itself will map $p$ to any other leaf of $T_u$, since they are
at different levels of $T_u$.  Therefore, we must give $p$ a
different isomorphism class label from other leaves of $T_u$.

For any other neighbor $w$ of $u$, $w$ lies at level $i-1$, and, by induction,
it is labeled with an integer isomorphism class label for level $i-1$.
The isomorphism classes reflect the stronger constraints where, if $w$
and $w'$ are two nodes at level $i-1$ that have the same label,
$T_w$ and $T_{w'}$ are isomorphic as rooted quotient-labeled PC trees.
If, together in $T$ and $T'$, there are $k$ distinct isomorphism equivalence 
classes for internal nodes at level $i-1$, they are labeled with integers 
between 1 and $k$, where 1 denotes that a vertex is a leaf at level $i-1$.
For each neighbor of $u$, let $e_w$ denote the integer label from 
$\{-1, 0, 1, \ldots, k\}$ assigned to $w$.

We now strengthen the inductive step of Section~\ref{sect:AHU74} to make
the stronger induction hypothesis go through for level $i$.
The canonical encoding of the quotient at $u$ given in Section~\ref{sect:canon}
assigns a unique tuple
to each neighbor $w$ of $u$; to this tuple we simply prepend $e_w$
to the tuple generated for $w$.  The $e_w$ labels on neighbors
enforce the constraint
that $u$ and $u'$ can get the same tuple only if there is a bijection
$\pi$ from neighbors of $u$ to neighbors of $u$ such that $T_w$ and $T_{\pi(w)}$
are isomorphic as quotient-labeled
trees.  The rest of the tuple forces the constraint that they can get
the same tuple only if there exists such a $\pi$ that is also an
isomorphism from the quotient at $u$ to the quotient at $u'$, as
in Section~\ref{sect:canon}.  Conversely, after ordering the tuples
in the canonical way described in Section~\ref{sect:canon}, it is clearly
sufficient for $T_u$ and $T_{u'}$ to be isomorphic as quotient-labeled
trees for $u$ and $u'$ to be assigned the same tuple.

Replacing the tuples with integer codes from 1 to $k$, where $k$ is
is the number of distinct tuples at level $i$ completes the induction
step.  

Therefore, $T$ and $T'$ are isomorphic quotient-labeled PC trees if and only if,
after rooting them at their centers and performing this algorithm,
the roots get assigned the same integer label.

\subsection{Time bound}

\begin{theorem}\label{thm:circIsoTime}
Given the sparse representations of matrices $M$ and $M'$,
it takes $O(size(M))$ time 
either to determine that neither is a circular-ones matrix, or else to
determine whether they are isomorphic.  Given succinct representations
of two circular-ones matrices, this problem takes $O(n)$ time to solve.
\end{theorem}
\begin{proof}
It takes $O(size(M))$ time to determine
whether they have the same number of 1's, by counting
1's in the two matrices in parallel.  If so, $size(M) = size(M')$.
If the standard sparse representations of the matrices is given,
it takes $O(size(M))$ time to determine whether they are circular-ones
matrices.  If neither is, the claim is satisfied.  If only one is,
they are not isomorphic.  Otherwise, it takes $O(size(M))$ time to convert them to 
the succinct representations.  From the succinct representations we can
compute the two quotient-labeled PC trees, as described above.

Therefore, by Lemma~\ref{lem:PCIsomorphism}, it suffices to show that the 
quotient-labeled PC-tree isomorphism algorithm 
can be implemented to run in $O(n)$ time.

\vspace{.2in}

{\bf Proposition 1:}  {\em Summing, for every level $i$, the number
of nodes at level $i-1$ plus the number of rows in quotients at level
$i$ gives a number that is $O(n)$.}  This just counts the number of
nodes in the tree plus the number of rows in the quotients.  Each row
of a matrix projects to just one row of a quotient.

{\bf Proposition 2:}  
{\em At level $i$, the sum of lengths of the tuples is at most proportional 
to the number of nodes at level $i-1$ plus the number of rows in quotients at 
level $i$.}  This is because one tuple is generated for each neighbor $w$
of $u$ contains an integer equivalence class label $e_w$,
and an encoding of a set of rows of the quotient at $u$.  The encoding of
each row in the quotient only appears in one of the tuples for neighbors.

{\bf Proposition 3:}
{\em At level $i$, the maximum integer in any tuple is bounded
by the number of nodes at level $i-1$ plus the number of rows
in quotients at level $i$.}  The integer equivalence classes at level
$i-1$ are assigned consecutive numbers, starting at 2, by sorting the
tuples lexicographically, and giving the same integer to two consecutive
tuples that are identical, and giving an integer that is one higher
than its predecessor's if it differs from its predecessor.
Each row of a quotient at level $i$ maps to only one element of 
a tuple generated at level $i$.

{\bf Proposition 4:}  A radix sort of a set of tuples of integers takes 
time proportional to the sum of lengths of the tuples plus the 
size of the range of integer values occurring in the tuples~\cite{CLRS01}.  

\vspace{.2in}

The tuples at P nodes must be sorted lexicographically.  
Number the nodes at level $i$ in any order to assign them
identification numbers, or {\em I.D. numbers}.  The maximum number label of one
of these nodes is at most the number of nodes at level $i-1$.
To each tuple for a child of a P node, prepend the I.D. number of
the parent.  Sort all tuples for P nodes at level $i-1$ in a single
lexicographic sort.  Since the I.D. number of the parent is the
major sort key, this groups all tuples of children of a P node
together, in lexicographic order, giving, for each P node,
one lexicographically sorted list of tuples for its children.
By Propositions 2, 3, and 4, the time for sorting all lists of tuples for
children of P nodes conforms to the measure given in Proposition 1.

The order of tuples of children of a Q node are already given by
the canonical procedure for generating them, as described above.

We must also sort the set of lists of tuples at level $i$ in order to generate
the equivalence class numbers for the nodes at level $i$.  A list
of tuples can be represented by a simple tuple of integers by appending a 
special separator, -2, to each tuple, and then concatenating them.
This change of representation does not affect the lexicographic order
of the lists, but turns them from lists of tuples to lists of integers
to make it easier to see that they can be radix sorted.  The addition
of the separators increases the range of values by $O(1)$.
By Propositions 2, 3, and 4, the time for sorting the set of lists of 
tuples for children of P nodes conforms to the measure given in Proposition 1.

Assigning integer equivalence-class labels to the lexicographically sorted
set of lists of tuples trivially takes time proportional to the sum of lengths
of lists of tuples, which, by Proposition 2, conforms to the measure of 
Proposition 1.

All of these steps conform to the measure of Proposition 1, so,
by Proposition 1,  they take $O(n)$
time over all iterations of the induction step.

We must also bound the time to choose, from the $2k$ possible choices
of a list of tuples at a C node, one that is earliest in lexicographic
order.  Generate two lists, one for each cyclic ordering, starting
at an arbitrary node for each.  Turn each of the lists from a list
of tuples to a list of integers, using the separators, as
described above.  Then apply the linear-time algorithm
of~\cite{Shi81} to find the cyclic rotation of each list that
is earliest in lexicographic order.  Of these two resulting lists,
choose the one that is earlier in lexicographic order.
\end{proof}

\section{Helly circular-arc graphs}\label{sect:Helly}

Every interval model has the Helly property~\cite{FulkGross65}. However,
unlike interval models, circular-arc models may fail to have the 
Helly property.
Figure~\ref{fig:nonHCA} gives a circular-arc model of a graph where  
the arcs that make up a clique, $\{A,B,C\}$, do not have a common intersection.

A graph is a Helly circular-arc graph if it admits at least one
circular-arc model that has the Helly property.
It is easily verified that there is no circular-arc model of the
graph of Figure~\ref{fig:nonHCA} where $A$, $B$, and $C$ have a common
intersection point.  This illustrates that
the Helly circular-arc graphs are a proper subset of the circular-arc graphs.

\begin{figure}
\centerline{\includegraphics[]{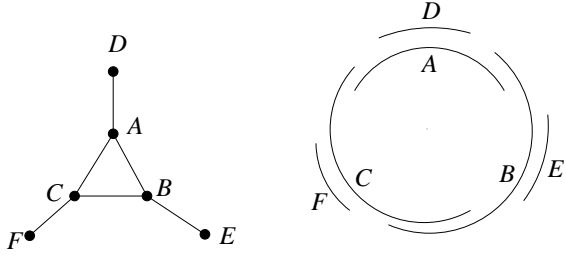}}
\caption{A non-Helly circular-arc graph.}\label{fig:nonHCA}
\end{figure}

Not every circular-arc model of a Helly circular-arc graph has the Helly
property.  Removing arcs $\{D,E,F\}$ from the model of Figure~\ref{fig:nonHCA}
yields a model for the complete graph on three vertices that does not have
the Helly property.  However, this graph is a Helly circular-arc graph;
it is easy to represent this graph with three copies of the same arc, which
has the Helly property.

The strategy of our
algorithm for finding whether two Helly circular-arc graphs are isomorphic is 
related to that of Lueker and Booth for finding whether two interval graphs are
isomorphic~\cite{LB79}; both algorithms use the 
clique matrices of the graphs.  The main new challenge 
involves correctly computing the clique matrix.

We consider adjacency-list representations
of two graphs $G$ and $G'$.  Recall that we assume that
the numbers of vertices and edges in both graphs are the same.
For each graph,
we obtain a Helly circular-arc model 
or determine that none exists, in $O(n+m)$ time,
using an existing algorithm for this problem~\cite{JLMSS11}.  If exactly
one of them is a Helly circular-arc graph, we determine that they are non-isomorphic.
If both of the input graphs fail to be Helly circular-arc
graphs, we reject them for failing to meet the precondition, even
though they may be isomorphic.

If both graphs are Helly circular-arc graphs, we
use the Helly circular-arc models to find succinct representations of 
circular-ones orderings of their clique matrices.  This involves some
complications not present in the corresponding problem on interval graphs,
which we show how to get around in $O(n)$ time, below.
Once we have succinct representations of the
clique matrices, we use the straightforward fact that two graphs are 
isomorphic if and only if their clique matrices are isomorphic
(Lemma~\ref{lem:cliqueIsomorphism}, below).
This reduces the
problem to isomorphism of circular-ones matrices, which we have solved 
in $O(n)$ time in Section~\ref{sect:PCiso}.  The total time is $O(n+m)$.

If the two graphs are circular-arc graphs and the inputs are circular-arc
models, we use the $O(n)$ algorithm of~\cite{JLMSS11} to find a Helly 
circular-arc model for each of them or determine that it is not a Helly 
circular-arc graph.   We then proceed as in the case where adjacency-list
representations of two graphs are given, but take a total of $O(n)$
time, rather than $O(n+m)$.

The main result of this section is the following theorem:

\begin{theorem}\label{thm:hcaIsomorphism}
Given the adjacency-list representations of two graphs, $G$ and $G'$, where
$G$ has $n$ vertices and $m$ edges, it takes $O(n+m)$
time either to determine that neither is a Helly circular-arc graph, or
else to determine whether they are isomorphic.  Given circular-arc models
of two circular-arc graphs, this problem takes $O(n)$ time to solve.
\end{theorem}

If $n'$ and $m'$ are the number of vertices and edges of $G'$, it
takes $O(n+m)$ time to determine whether $n = n'$ and $m = m'$, by
counting these elements in the two graphs in parallel.
If this is not the case, then they are not isomorphic.  Therefore,
we may assume henceforth that $n = n'$ and $m = m'$.

\begin{lemma}[\cite{Chen96}] \label{lem:cliqueIsomorphism}
Two graphs are isomorphic if and only if their
clique matrices are isomorphic.
\end{lemma}
\begin{proof}
A graph isomorphism maps maximal cliques to maximal cliques,
so it defines an isomorphism of their clique matrices.
Conversely, since two vertices are adjacent if and only if they are
in a common clique, an isomorphism from the clique matrix of one
graph to that of another defines a graph isomorphism.
\end{proof}

In a circular-arc model of a graph, let an {\em intersection segment}
be a place where a counterclockwise endpoint of an arc is followed immediately by
a clockwise endpoint of an arc in the model in the clockwise direction;
the intersection segment is
the region between the two points.  In an interval model, 
a set of arcs corresponds to a maximal clique if and 
only if it is the set of intervals containing an intersection segment.
Each intersection segment can be located and marked by a 
{\em clique point} lying in the segment.
A clique is a maximal clique if and only if it is the set of arcs that 
contain a clique point.  A consecutive-ones ordering
of the clique matrix can be obtained by making one column for each such
maximal clique, and putting a 1 in the column in each row corresponding
to an arc that passes through the region.

By analogy, in a Helly circular-arc model, we can place a clique point
in each intersection segment.
Because the model has the Helly property, a maximal clique
must be the set of arcs containing one of the clique points.  However,
in contrast to the case of interval models,
not all such sets of arcs must be maximal cliques.  
Moreover, the same clique may appear multiple times, as
the arcs containing two different clique points in widely separated
parts of the circle.  Figure~\ref{fig:HellyCliques}
gives an example.

\begin{figure}
\centerline{\includegraphics[]{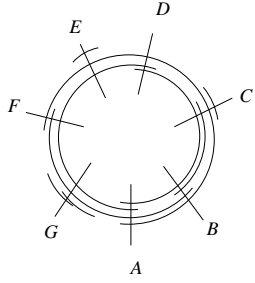}}
\caption{A Helly model where not every intersection segment corresponds
to a maximal clique.  The sets of arcs that contain the intersection segments
at points D and F are subsets of the one that contain the intersection
segment at A.}\label{fig:HellyCliques}
\end{figure}

To get a clique matrix for the graph, we must eliminate
redundant clique points and clique points that do not represent maximal
cliques from this Helly circular-arc model.   Suppose we accomplish this.
We can obtain a succinct representation of a circular-ones ordering
of its clique matrix as follows.  Number the arcs and label each arc's
two endpoints with its arc number.  Create an array of $n$ buckets, one
for each vertex.  Number the cliques by numbering the clique points
that have not been eliminated in order around the circle.  This is
a circular-ones ordering of the clique matrix.  

It suffices, for each bucket $i$, to
store the counterclockwise-most and clockwise-most clique number 
of each arc in order to get the succinct representation of the circular-ones
ordering of the clique matrix.  The counterclockwise points that correspond to 
clique $j$ are those in the maximal consecutive block of counterclockwise points
that lie immediately counterclockwise from clique point $j$.  
The clockwise points that correspond to
clique $j$ are those in the maximal consecutive block of clockwise 
endpoints immediately clockwise from clique point $j$.  This can be recorded 
in each bucket, giving the succinct representation of a circular-ones 
ordering of the clique matrix.

We now
describe how to eliminate redundant clique points and clique points
that do not correspond to maximal cliques.  Since we have placed one clique
point at each intersection segment, and each intersection segment contains
the counterclockwise point of an arc and a clockwise point of an arc, we have
placed at most $n$ clique points. A preliminary version of this procedure
appeared in \cite{LMSS08}.

For any point $p$ on the circle of a circular-arc model, denote by 
${\cal C}(p)$ the family of arcs of the model that contain $p$.  
Given two points $p_1$ and $p_2$ on the circle of a circular-arc model,
let us say $p_1$ {\em dominates} $p_2$ if ${\cal C}(p_2) \subseteq {\cal C}(p_1)$.
Let ${\cal A}$ be a circular-arc model and $P = \{p_1, \ldots, p_k\}$ be a 
set of points of the circle on which ${\cal A}$ resides,
where $(p_1, \ldots, p_k)$ is the clockwise order in which $P$ appears in
a traversal of the circle, starting at an arbitrary point on the circle.  
Let us say that $P' \subseteq P$ is a 
\emph{$P$-dominating} set if every point in $P \setminus P'$ is dominated by 
some point in $P'$.  Any minimal set of dominating points, with respect to containment, among the set $P$ of at most 
$n$ clique points we have placed on the circle, is a non-redundant set of
clique points.  We solve the following more
general problem in $O(n+|P|)$ time:

\begin{itemize}
\item Given a set $P$ of points on the circle of a (not necessarily Helly) 
circular-arc model, find a minimal $P$-dominating set.
\end{itemize}

If the model is Helly and $P$ consists of one point per intersection segment, 
a minimal P-dominating set gives the columns of the clique matrix.

The \emph{ascending semi-dominating sequence} of $P$ is the subset 
$SD^+(P) = \{p_i \in P \mid {\cal C}(p_i) \not \subseteq {\cal C}(p_j)$
for all $p_j \in P$ such that $1 \leq i<j\leq k \}.$  In 
other words, $SD^+(P)$ contains the points $p_i \in P$ that are not 
dominated by any later point in $P$.  Similarly, the 
\emph{descending semi-dominating sequence} of $P$ is the subset 
$SD^-(P) = \{p_j \in P \mid  {\cal C}(p_j) \not \subseteq {\cal C}(p_i)$
for all $p_i \in P$ such that $1 \leq i < j \leq k\}.$
The following lemma reduces the problem of finding a minimal $P$-dominating
sequence to that of finding $SD^+$ and $SD^-$.

\begin{lemma} \label{thm:SD}
Let ${\cal A}$ be a circular-arc model and $P$ be a set of points on it.  
Both $SD^-(SD^+(P))$ and $SD^+(SD^-(P))$ are minimal $P$-dominating sequences.
\end{lemma}

\begin{proof}
We only prove that $P^* = SD^+(SD^-(P))$ is a minimal $P$-dominating 
sequence.  The proof for $SD^-(SD^+(P))$ can be obtained by taking the 
reverse of ${\cal A}$.  Let $P = \{p_1, \ldots, p_k\}$ be points on 
the circle where $(p_1, \ldots, p_k)$ 
is the order in which $P$ appears in a clockwise traversal of the circle.  
We first prove that $P^*$ is in fact a $P$-dominating sequence.

By definition, every point $p_j \in P \setminus SD^-(P)$ is dominated 
by some point $p_i \in P$ for some $1 \leq i < j$.  If $i$ is the minimum 
element in $\{1, 2, \ldots, k\}$
such that $p_i$ dominates $p_j$, then no point $p \in \{p_1, \ldots, 
p_{i-1}\}$ can dominate $p_i$; otherwise $p$ would dominate $p_j$, 
contradicting the minimality of $i$.  Therefore, every point in 
$P \setminus SD^-(P)$ is dominated by some point in $SD^-(P)$.  We can 
apply a symmetric arguments for $SD^-(P)$ and $P^*$ to conclude that every 
point in $SD^-(P) \setminus P^*$ is dominated by some point in $P^*$.  
Since domination is a transitive relation, every point of $P$ is also 
dominated by some point in $P^*$, i.e., $P^*$ is a $P$-dominating sequence.

We now show that $P^*$ is minimal.  Observe that it is enough to show 
that if a point $p_i$ is dominated by a point $p_j \in P^*$, where
$p_j \neq p_i$, then $p_i \not\in P^*$.
This will imply that no point of $P^*$ dominates any other point of $P^*$, 
so no point of $P^*$ can be removed from it to yield a smaller dominating set.
If $j < i$ then $p_i \not \in SD^-(P)$, hence 
$p_i \not \in P^*$.  If $j > i$, then since  $p_j \in P^*$,
it follows that $p_j \in SD^-(P)$, and
we obtain again that $p_i \not \in SD^+(SD^-(P)) = P^*$.
\end{proof}

The algorithms for finding $SD^+$ and for finding $SD^-$ are symmetric. 
We describe the one to 
find $SD^+$.  The algorithm works by induction on $i$ to
find $SD^+(P_i)$, where $P_i = \{p_1, p_2, \ldots, p_i\}$.
That is, we find those points of $P_i$ that
are not dominated by any later points of $P_i$.  

By induction, assume we have the following partition of $SD^+(P_i)$ at the
end of step $i$:

\begin{itemize}
\item $D_i$:  points in $SD^+(P_i)$ that are already known to be in $SD^+(P)$.
\item $Q_i$:  remaining points in $SD^+(P_i)$; these are points that
are not dominated by any later point in $P_i$, but that may or may not be
dominated by points in $\{p_{i+1}, p_{i+2}, \ldots, p_k\}$.  
\end{itemize}

It is easy to see that it will follow that when
$i = k$, we get that $SD^+(P_k) =
SD(P) = D_k \cup Q_k$, and this solves the problem.

We begin with $SD^+(P_1) = \{p_1\}$, where $D_1 = \emptyset$ and $Q_1 = \{p_1\}$.

In step $i+1$, we obtain $Q_{i+1}$ and $D_{i+1}$ from $Q_i$ and $D_i$
as follows.  We remove points from $Q_i$
and insert them in $D_i$ if they pass a test that shows that they
cannot be dominated by any later point in $P$, including $p_{i+1}$.
The addition of these points to $D_i$ gives $D_{i+1}$.  We discard other 
points from $Q_i$ that are dominated by $p_{i+1}$.  We then add 
$p_{i+1}$ to $Q_i$.  This gives $Q_{i+1}$.

The test of whether a point $q$ moves from $Q_i$ to $D_{i+1}$ consists
of determining whether it is contained in the arc $B_{i+1}$
that does not contain
$p_k$, has its clockwise endpoint in $[p_i,p_{i+1})$, and
among all such arcs, extends farthest in the counterclockwise direction.
(See Figure~\ref{fig:francisco}.)
If $q$ is contained in $B_{i+1}$, then, since $B_{i+1}$ 
does not contain any point from $\{p_{i+1}, p_{i+2}, \ldots, p_k\}$, $q$
cannot be dominated by any point in this set, and since $q \in SD^+(P_i)$,
it it is not dominated by any later
point in $P_i$.  Therefore, it is a member of $SD^+(P)$, and can be
moved to $D_{i+1}$.

\begin{figure}
\centerline{\includegraphics[]{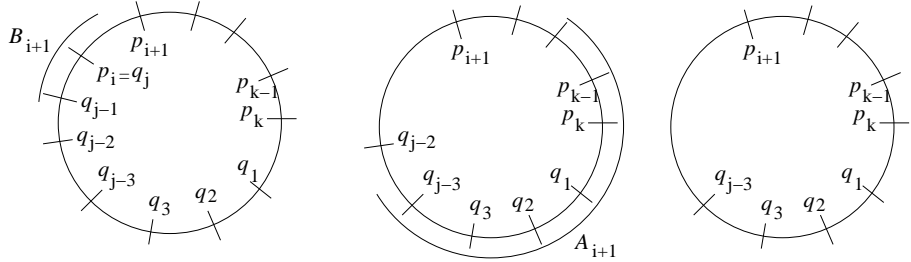}}
\caption{Computing $D_{i+1}$ and $Q_{i+1}$ from $D_i$ and $Q_i$.  Out of all
arcs that have their clockwise endpoints in $[p_i,p_{i+1})$ and
do not contain $p_k$,
$B_i$ is the one that extends farthest counterclockwise.
Elements of $Q_i$ that are contained in $B_{i+1}$ cannot be dominated
by any point in $\{p_{i+1}, p_{i+2}, \ldots, p_k\}$, hence they
are moved to $D_i$, yielding $D_{i+1}$.
Out of all arcs that contain $p_k$ but do not contain $p_{i+1}$, 
$A_{i+1}$ is the one that extends farthest clockwise.  Elements
still in $Q_i$ that are not contained in $A_{i+1}$ are dominated
by $p_{i+1}$, hence discarded from $Q_i$.  Then $p_{i+1}$ is added
to what remains of $Q_i$, yielding $Q_{i+1}$.}\label{fig:francisco}
\end{figure}

We implement $Q_i$ as a stack $(q_1, q_2, \ldots, q_j)$, where $q_j$
is the top of the stack.
Note that the set of points that get moved to $D_i$ are consecutive
at the top of the stack.  To move them, we pop the stack until
we reach an element not in $B_{i+1}$, and move the popped elements
to $D_i$.

A point $q'$ that is still in $Q_i$ fails to be dominated by $p_{i+1}$ 
if and only if it is contained in some arc 
that does not contain $p_{i+1}$.   All arcs in this set contain $p_k$,
since otherwise, $q'$ would already be identified as a member of $D_j$
for some $j \leq i+1$.

Of all arcs that exclude $p_{i+1}$ but contain $p_k$, let $A_{i+1}$ be 
the one whose clockwise endpoint extends farthest
clockwise from $p_k$.  (See Figure~\ref{fig:francisco}.)
Since $A_{i+1}$ is the arc in this set that covers the most 
members of $Q_i \setminus D_{i+1}$, it follows that the points
of $Q_i$ that are dominated by $p_{i+1}$ are those that are
not contained in $A_{i+1}$.

Note that the ones that are dominated by $p_{i+1}$ are again consecutive
at the top of the stack, so we pop the stack until we reach
an element that is contained in $A_{i+1}$, and discard the popped elements.

Since $p_{i+1}$ belongs in $Q_{i+1}$, we obtain our stack for $Q_{i+1}$
by pushing $p_{i+1}$ to what remains of the stack for $Q_i$.

For the time bound, note that we may find $B_{i+1}$ for each 
$i \in \{1, 2, \ldots k\}$ by traversing $[p_i,p_{i+1})$ comparing
the counterclockwise endpoints of arcs whose clockwise endpoints
are in this interval and do not contain $p_k$.  Over all $i$, this expends
$O(1)$ time on each arc, so it takes $O(n)$ time.  

To find $A_{i+1}$ for all $i \in \{1,2, \ldots k\}$, we start just
counterclockwise from $p_k$
and traverse the circle counterclockwise, keeping track of the {\em best}
arc so far.  The best arc is initially null.   When we reach
an arc $A$, we check whether $A$ contains $p_k$, and, if so,
whether it extends farther clockwise than the best arc so far.
If so, $A$ becomes the best arc so far.
Each time we reach a point $p_i$, we
record the best arc so far as $A_{i+1}$.  
Over all $i$, this also expends $O(1)$ time 
on each arc, for a total of $O(n)$ time.  

The management of the
stack implementing $Q_i$ takes $O(n)$ time over all steps, since
each point is pushed once, and when points are popped, they are
consecutive at the top of the stack.

\section{Isomorphism of $\Gamma$ circular-arc graphs, convex-round graphs, and proper circular-arc graphs} \label{sect:Gamma}

In this section we show that the circular-ones matrix isomorphism 
test of Section~\ref{sect:matrixIsomorphism}
can be used to test isomorphism of $\Gamma$ circular-arc graphs 
and convex-round graphs, using results of Chen~\cite{Che00}. From
this we get a new algorithm for isomorphism of proper circular-arc graphs.

\begin{theorem}
Given adjacency lists of two graphs, it takes $O(n+m)$ time to either determine that the graph are not
$\Gamma$ circular-arc graphs or to determine whether they are isomorphic.
\end{theorem}

\begin{theorem}
Given adjacency lists of two graphs, it takes $O(n+m)$ time to either determine that the graph are not
convex-round graphs or to determine whether they are isomorphic.
\end{theorem}

A graph is 
a $\Gamma$ circular-arc graphs if its augmented adjacency matrix has the 
circular-ones property. Chen~\cite{Che00} showed that two $\Gamma$ circular-arc graphs
are isomorphic if and only if their augmented adjacency matrices are isomorphic.

Given adjacency-list representations of two graphs, it takes $O(n+m)$ time 
to determine whether their augmented adjacency matrices have the
circular-ones property~\cite{Gol80,HsuMcC03,Tuck71}.  If they both do, then it takes 
$O(n+m)$ time, using the isomorphism test of 
Section~\ref{sect:matrixIsomorphism} to 
determine whether they are isomorphic.  If a circular-one ordering of the adjacency matrices of the two
graphs are given using succinct representations, the test takes $O(n)$
time.

Convex-round graphs are complements of $\Gamma$ circular-arc 
graphs~\cite{BHY00}. The adjacency matrix of a convex-round graph has the 
circular-ones property.  Chen~\cite{Che00} showed that two convex-round graph
are isomorphic if and only if their adjacency matrices are isomorphic. Therefore,
we use the same technique to get the same bounds for testing isomorphism of
convex-round graphs as we do for testing isomorphism of $\Gamma$
circular-arc graphs.

Since every proper circular-arc graph is a $\Gamma$ circular-arc 
graph~\cite{Tuck71}, we can use the same algorithm
also for an isomorphism test of proper circular-arc graphs. This gives 
a new $O(n+m)$ isomorphism algorithm for proper circular-arc graph.

The $O(n+m)$ time bound
is optimal if the input graphs are given by adjacency lists, but it is not 
optimal if they are given by proper circular-arc models.
The algorithm of Lin et al.~\cite{LSS08} for the problem solves it in 
$O(n)$ time if the circular-arc models
are given.  We show how to achieve an algorithm with the same
time bound. We need to find succinct representations of circular-one
arrangements of the augmented adjacency matrix of the given proper circular-arc models.

Let ${\cal A}$ be a circular-arc model of $G$.
If ${\cal A}$ is not a proper circular-arc model, then we can convert
it in $O(n)$ time to such a model \cite{Nussbaum08}.
The model ${\cal A}$ can be changed 
in $O(n)$ time, such that no two arcs cover the circle together~\cite{KaplanNussbaumCert09, LS08}, and the model remains proper. 
After changing the model this way, we index the 
vertices of $G$ according to the clockwise order of their counterclockwise
endpoints, 
starting at an arbitrary endpoint. Since there are no arc containment or 
pair of arcs that cover the circle in ${\cal A}$, this indexing gives a 
circular-ones arrangement of the augmented adjacency matrix of $G$.

To find the last 1 entry in the row of each vertex of $G$ according to this
indexing, in $O(n)$ time, we go clockwise around the circle once, starting at an
arbitrary counterclockwise endpoint, which belongs to a vertex $v$.
If the next endpoint we encounter is a clockwise endpoint 
of a vertex $u$, then the last 1 in the row of $u$ is in the column of $v$. If the next 
endpoint we encounter is a counterclockwise endpoint of a vertex $u$, we set 
$v = u$. We end this traversal when we get back to the start point. The 
first 1 of every row of a vertex in $G$ is found symmetrically.

We conclude that it takes $O(n)$ time to find a succinct representation
of a circular-ones ordering of the augmented adjacency matrix of a graph
from a proper circular-arc model of the same graph.
Once we have succinct representation of the augmented adjacency matrices
of $G$ and $G'$ we can test them for isomorphism in $O(n)$ time.

\begin{theorem}
Given adjacency lists of two graphs, it takes $O(n+m)$ time to either determine that the graph are not
proper circular-arc graphs or to determine whether they are isomorphic. Given two circular-arc models of
circular-arc graphs, the same task takes $O(n)$ time.
\end{theorem}

\section{Hsu's algorithm for circular-arc graphs isomorphism} \label{sect:hsu}

In this section we give an example of two circular-arc graphs that are not isomorphic, 
but the algorithm of Hsu \cite{Hsu95} determines that they are. We begin with 
few definitions from this paper. Let $G$ be a circular-arc graph, without 
universal vertices, and without any pair of vertices $v$ and $u$ such that $N[v]=N[u]$. 
A \emph{normalized model} of $G$ is a circular-arc model of the graph, such that for 
every two arcs $v$ and $u$: (1) if $N[u] \subseteq N[v]$ then the arc of $v$ contains 
the arc of $u$; (2) if every $w \in V \setminus N[v]$ satisfies $N[w] \subseteq N[u]$ 
and every $w' \in V \setminus N[u]$ satisfies $N[w'] \subseteq N[v]$ then the arcs 
of $u$ and $v$ cover the circle together. Every circular graph without universal arcs 
and without any pair of vertices with the same neighborhood has a normalized 
model \cite{Hsu95}.

Let $\cal A$ be a normalized model of $G$. We get the \emph{associated chord model} 
of $\cal A$ by replacing the arcs of $\cal A$ with chords. Let $\cal A'$ be an associated 
chord model of $G$. The chord model $\cal A'$ represents a \emph{circle graph} $G_C$, whose 
vertex set is the same vertex set as of $G$, and two vertices are adjacent 
if and only if their chords in $\cal A'$ intersect. Although there might be more 
than one unique associated chord model for $G$, the graph $G_C$ is unique. 
Hsu \cite{Hsu95} defined a type of chord model called \emph{conformal model} 
for $G_C$. The chord model $\cal A'$ is a \emph{conformal model} of $G_C$. Note 
that we do not repeat the definition of conformal model here, we just give 
an example for one such model. We do not require the definition for our purposes.

The origin of the mistake in Hsu's algorithm is the statement \emph{``To 
test the isomorphism between two such circular-arc graphs $G$ and $G'$, 
it suffices to test whether there exist isomorphic conformal models for 
$G_C$ and $G'_C$''} \cite[Section 9]{Hsu95}, where ``such circular-arc 
graphs'' refer to circular-arc graphs for which both $G_C$ and its 
complement are connected.

Consider the graphs $G$ and $G'$ in Figure \ref{fig:hsugraphs}. It is easy to 
see that $G$ and $G'$ are not isomorphic, since they have different number 
of edges. Normalized circular-arc models of both graphs are given in 
Figure \ref{fig:hsumodels}. From the circular-arc models, we can see 
that the chord model in Figure \ref{fig:hsuchords} is an associated 
chord model for both graphs, and hence a conformal model of both 
$G_C$ and $G'_C$. The graphs $G_C = G'_C$ are connected and so are 
their complements. We conclude that the statement above is wrong, and 
the algorithm of \cite{Hsu95} falsely finds that $G$ and $G'$ are isomorphic.

\begin{figure}
	\center
    \includegraphics[scale=0.5]{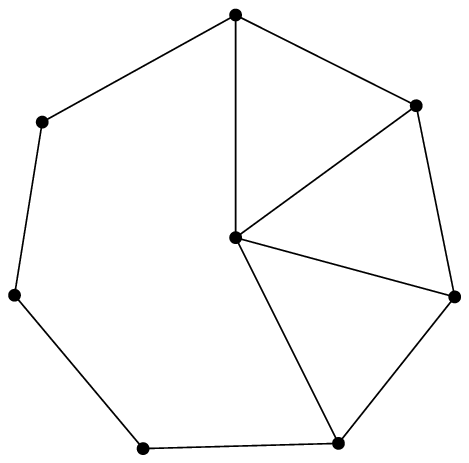}
    \hspace{2cm}
    \includegraphics[scale=0.5]{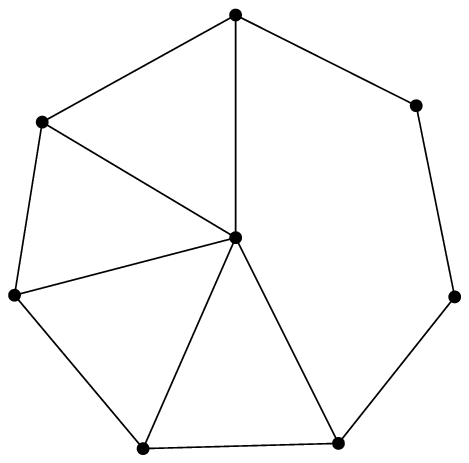}
    \caption{Two circular-arc graphs $G$ and $G'$. It is easy to see that the two graphs are not isomorphic to each other, since they have different number of edges.}
    \label{fig:hsugraphs}
\end{figure}

\begin{figure}
	\center
    \includegraphics[scale=0.5]{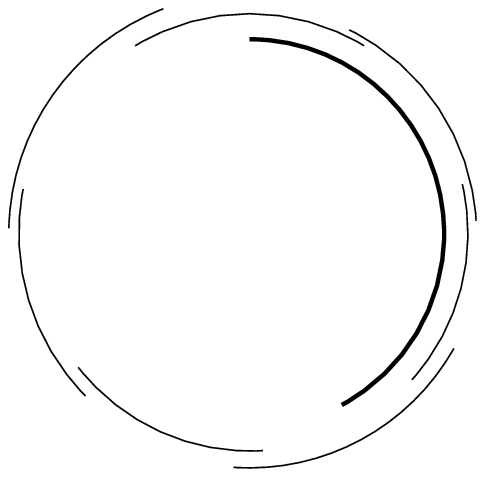}
    \hspace{2cm}
    \includegraphics[scale=0.5]{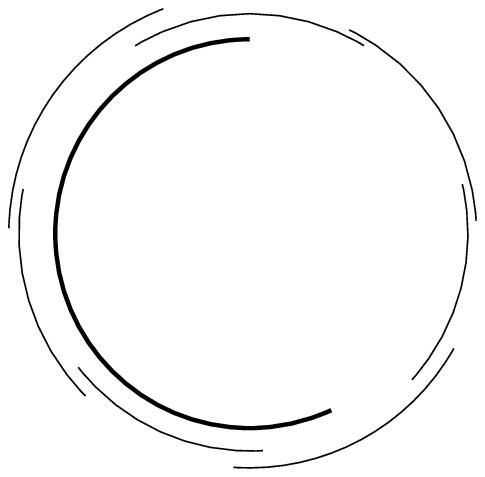}
    \caption{Normalized circular-arc models of $G$ and $G'$. The two models share seven arcs in common, and the eighth arc (\emph{bold}) is flipped between the two models.}
    \label{fig:hsumodels}
\end{figure}

\begin{figure}
	\center
    \includegraphics[scale=0.5]{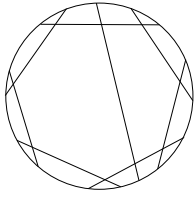}
    \caption{The associated chord model of both $G$ and $G'$.}
    \label{fig:hsuchords}
\end{figure}

Hsu \cite{Hsu08} noted that the isomorphism of conformal models of 
$G_C$ and $G'_C$, which the algorithm of \cite{Hsu95} produces, 
does give a mapping between vertices of $G$ and $G'$, if $G_C$ is 
inseparable with respect to modular decomposition. However, we do 
not know how to handle the case when this condition is not satisfied.
We note that the algorithm of \cite{Hsu95} works correctly for isomorphism
of circle graphs. With the $O(n^2)$ time recognition algorithm for circle graphs \cite{circrec},
the isomorphism test takes $O(n^2)$ time
if the graphs are given as adjacency matrices. The recent $O((n + m)\alpha(n + m))$
circle-graph recognition algorithm \cite{GPTC11} leads to the same running time for
circle-graph isomorphism, where $\alpha(\cdot)$ is the inverse Ackermann function.
If chord models are given as an input, then the running time of the
isomorphism test can be reduced to $O(n+m)$ using techniques similar to those used
in \cite{LSS08} and in our paper.


\end{document}